\documentclass{amsart}

\usepackage{ amsthm,   amsmath,   amssymb,   enumitem,   geometry,   caption,   graphicx,   hyperref,   fontenc,   inputenc}
\usepackage{dsfont}
\usepackage{tikz}
\usepackage{tikz-network}
\usepackage{dsfont}
\usepackage{amssymb,amsmath,graphicx,mathtools}
\usepackage{bbold}
\newcommand{\cc}{{\tt g}}

\newcommand{\mA}{{\bf A}}

\newcommand{\mT}{{\bf T}}
\newcommand{\nvec}{{\bf n}}
\newcommand{\mvec}{{\bf m}}
\theoremstyle{plain}
\newtheorem{theorem}{Theorem}
\newtheorem{lemma}[theorem]{Lemma}
\theoremstyle{definition}

\newtheorem{estimate}[theorem]{Estimate}

\begin{document}

\title[Non-equilibrium stationary state for a one-dimensional stochastic wave equation]{Perturbation theory for a non-equilibrium stationary state of a one-dimensional stochastic wave equation}


\author[G.~Guadagni]{Gianluca Guadagni}
\address{Applied Mathematics, University of Virginia, Charlottesville,
  Va, 22904} 
\email{gg5d@virginia.edu}

\author[L.E.~Thomas]{Lawrence E. Thomas}
\address{Mathematics, University of Virginia, Charlottesville, Va, 22904}
\email{let@virginia.edu}







\maketitle 

\begin{abstract}
 We address the problem of constructing a non-equilibrium stationary state for a one-dimensional  stochastic Klein-Gordon wave equation 
with non-linearity, using perturbation theory.  The linear theory is reviewed, but with the linear equations of motion including an additional potential term which emerges  in the renormalization of the perturbation expansion for the state corresponding to the non-linear equations of motion.  The potential is the solution to a fixed point equation.  Low order terms in the expansion for the two-point function are determined.  
\end{abstract}



\begin{section}{Introduction}

The goal of this article is to introduce a perturbation theory for a putative non-equilibrium stationary state of a one-dimensional stochastic Klein-Gordon wave equation 
with non-linearity, with the field supported on a ring.  The field is weakly coupled to effective thermal baths of differing temperatures that provide  both dissipation and  noise. The stationary state for the linear dynamics is reviewed, but with the field equations of motion including an additional potential term which arises from a renormalization in the perturbation expansion.   The additional potential term, which depends on the difference of temperatures, alters the perturbation 
and changes  the unperturbed state.         
      
The equations for the field $\phi_t(x)$ and its momentum $\pi_t(x)$ are given by
\begin{align}\label{equationsofmotion.1}
\partial_t\phi_t(x)&= \pi_t(x)\nonumber\\
\partial_t\pi_t(x)&= (\partial_x^2 - 1)\phi_t(x)-  g(\phi_t(x))- r(t) \cdot \alpha(x)
\nonumber\\
dr_{i}(t)&=  -\left(r_{i}(t)-\langle\alpha_i,\pi_t\rangle\right)dt +
\sqrt{T_i}d\omega_i(t),\,\,i=1,2. 
\end{align}
We assume that the $ \phi_t$ and $\pi_t$ are periodic functions in $x\in [0,2\pi]$.  
The non-linearity $g$
will be a polynomial in the field, particularly $g(z)= z^3$. The two components of $r_t = (r_1(t),r_2(t))^{\intercal} \in{\mathds R}^2$  model the heat baths. The last equations allow  energy 
in the field to dissipate into the baths and serve as a source of energy into the field from the two standard independent Brownian motions
$\omega_i(t)$, $i= 1,2$.  $T_1$ and $T_2$ are the bath temperatures.  

The $\alpha_i(x)$, $i=1,2$ are fixed distributions coupling
the field to the two thermal baths, and $\langle \alpha_i$ is the operation of integrating $\alpha$ against another function of $x$, giving out a scalar.
The $\alpha_i$'s are assumed to satisfy the conditions on their Fourier coefficients $\{\widehat{\alpha}_i(n)\}$:  that there exist  positive constants $c_0,c_1,c_2$ with
$c_1<c_2<1 $ 
  such that $c_1 \leq|\widehat{\alpha}(n)\cdot \widehat{\alpha}(n)|\leq c_2 (\widehat{\alpha}^*(n)\cdot \widehat{\alpha}(n))\leq c_0$.  The bounds  ensure that eigenvalues of operators associated with the linear problem are non-degenerate and that certain small denominators (resonances) are manageable.  
    $\delta$-functions at antipodal points on the circle are ruled out by this condition, but modifications of them are of course possible.  
 
The Klein-Gordon field can be regarded as a limit of harmonic or anharmonic chain models in the limit with the number of oscillators growing to infinity and
with rescaling.   
There is a voluminous literature on  harmonic and anharmonic chains and their thermodynamics.
 One of the earliest contributions was that of Rieder, Lebowitz, and Lieb \cite{RLL}, who considered a linear chain
of oscillators with heat baths at the ends.  They discovered the surprising but now familiar phenomena that the energy current is proportional to the difference of temperatures but independent of the length of the chain, and  that the energy density along the chain has a peculiar
profile, particularly near its ends.   A model consisting of a finite chain of anharmonic oscillators was considered by
Eckmann, Pillet, and Rey-Bellet \cite {EPR},  who showed existence of a thermodynamic stationary state with entropy production.
An article by Aoki, Lukkarinen and Spohn \cite{ALS} concerned  determination of steady state energy flow through an anharmonic chain of
oscillators, with the system modeled via a kinetic theory of colliding phonons. They arrive at a Fourier law for the current and find
 good agreement  with simulations. Their article contains an excellent review of the literature on energy transport in simple models.

For continuum models, as considered here,
McKean and Vaninsky established a Gibbs state equilibrium measure for nonlinear wave equations and showed its invariance
under a Hamiltonian flow \cite{MV}.
Stationary states for stochastic non-linear wave equations have been considered by many others,
for example, Barbu and Da Prato \cite{BDaP}. Their dynamics, however, include
a dissipative $-\pi(x)$ term in the $\pi$ equation and a driving cylindrical Weiner process.They show existence and uniqueness
of an invariant measure assuming boundedness  on the derivative of the non-linear interaction.  Gubinelli, Koch, and Oh established local existence in time for a non-linear stochastic
wave equation driven by white noise in two dimensions involving a time-dependent (rather than spatially dependent, as here) renormalization\cite{GKO}. 
Rey-Bellet and Thomas  showed global existence in time for the above equations of motion (\ref{equationsofmotion.1}) in spaces of low regularity, $\phi\in H^s$, $s>1/3$
, for $g=-G'$, $G$
a polynomial bounded below, but for smoother $\alpha$'s than employed here \cite{RBT}. Wang and Thomas showed ergodicity for
a model with bounded interaction and ultraviolet cut-off \cite{WT}.

Closer to our problem, but
again in a discrete setting, is the work of Bricmont and Kupiainen,\cite{BK},
who consider a three-dimensional system of anharmonic oscillators between two parallel planes where noise and dissipation occur, effectively at different
temperatures. Stationarity of a putative measure leads to a system of relations for correlations of the field (BBGKY hierarchy); the system is truncated,
leading to an implicit relation for correlation functions.  They then find a resulting Fourier law for heat flow. There is presumably a connection between this hierarchy approach and the perturbative approach presented here, which 
involves an implicit relation for a correlation function.  The connection between the approaches remains to be explored.

Section 2 is a review of the linear problem $g=0$ \cite{T1}, but with a potential term $v= v(x)$ included in the $\pi$ equation,  
$v$ arising  in the non-linear perturbation problem.  Analytic  properties of the field covariance, and especially 
its dependence on $v$, are determined.
The spectral theory of the Liouville operator associated with the linear equations of motion is described. The spectral properties of this operator are largely  determined
by the matrix operator of its drift term, $\mA_v$.
 The analysis  of $\mA_v$, which is not normal, closely follows  the self-adjoint case, but with caveats 
 that  the eigenprojections of $\mA_v$ need to be estimated,  and that there is near degeneracy of pairs of
eigenvalues causing small-denominator (resonance) issues.  Estimates on the eigenvectors and eigenvalues appear in the text
as needed, but their proofs are relegated to  the appendix. 

 Section 3 is an introduction to a formal
perturbation expansion with $\phi^3$ non-linearity in the $\pi$-equation of (\ref{equationsofmotion.1}), the expansion having 
a certain 1950-60's quantum field theory flavor to it.  The potential $v$
makes its appearance by necessity here in a renormalization issue and is shown to be the solution to a fixed-point problem. Only first and second order terms in the expansion are considered.

Much of the analysis of $\mA_v$ is equivalent to  that for  a non-Hermitian one dimensional Schr\"{o}dinger-like operator with complex finite rank interactions;  part of the 
 appendix treats this  operator.

\end{section}

\begin{section}{\bf Aspects of the linear problem, $g(\phi)=0$}

\begin{subsection}{\bf Covariance of the field}
  Define the matrix operator ${\mA}_v$, 
\begin{equation}\label{bA.eq}
{ \mA}_v= \left(\begin{array}{cccc}
     0&1&0&0\\
    \partial_{x}^{2}-1-v(x)&0&-\alpha_1(x)&-\alpha_2(x)\\
     0&\langle\alpha_1&-1 &0\\
     0&\langle\alpha_2 &0&-1
     \end{array}
\right),
\end{equation}
which is simply the linear operator part of the right side of the equations of motion but with an additional continuous potential $v=v(x)$ which will
arise in the renormalization issue with $g\neq 0$.  The matrix $\mA_v$ isn't normal, rather it being the sum of  a matrix similar to a skew-adjoint operator and a rank $2$ symmetric perturbation. It operates on complex valued $4$-vectors of the form, $e= (e_{\phi}(x),e_\pi(x),e_r)^{\intercal}$; $e_\phi$ and $e_\pi$ are suitably
integrable functions in $x$, $e_r\equiv (e_{r_1},e_{r_2})^{\intercal}$ has two scalar components.  Let
        \begin{align} 
        \langle f,e\rangle=  \langle f,e\rangle_{\mathcal{H}} = \int \left(f^*_{\phi}(x)e_\phi(x) +f^*_{\pi}(x)e_\pi(x)\right) dx + f^*_r\cdot e^{\phantom{*}}_r,
\end{align}
 which will  serve both as a pairing of test functions with distributions and as an inner product for a  Hilbert space ${\mathcal{H}}$.   The matrix ${\bf \sqrt{T}}$ is defined  
 \begin{equation}\label{tA.eq}
\sqrt{\bf T} = \left(\begin{array}{cccc}
     0&0&0&0\\
   0&0&0&0\\
     0&0&{\sqrt{T_1}} &0\\
     0&0 &0&{\sqrt{T_2}}
     \end{array}
\right)
\end{equation}
with $T_1, T_2$ the bath temperatures. 

 Let $f_{\nvec}, \lambda^*_{\nvec}$ and $e_{\nvec},\lambda_{\nvec}$ be  eigenvectors and their eigenvalues  of ${\mA}_v^{\dagger}$ and ${\mA}_v$ respectively,
 $ \mA_v^{\dagger}f_{\nvec}= \lambda^*_{\nvec}f_{\nvec}$,  $\mA_ve_{\nvec}= \lambda_{\nvec}e_{\nvec}$. The mode number $\nvec$ is in the set $\{(n,i): \,\, n\in {\mathds Z}, \,i\in \{1,2\}\}$ or $\{1,2\}$; the latter labels two exceptional cases of eigenvectors with eigenvalues near $-1$ for small $\alpha$'s, their weight principally on
 their $r$ components.
          Define the projections $\{P_{\nvec}\}$, also $4\times 4$ matrices, by
\begin{align}\label{projectiondef}
           P_{\nvec}= \frac{|e_{\nvec}\rangle\langle f_{\nvec}|}{\langle f_{\nvec},e_{\nvec}\rangle_{\mathcal{H}}}.                       
     \end{align}
   Bounds on these operators, their eigenfunction components  and eigenvalues uniform in $\nvec$ and $v$ are obtained in the appendix,
    Lemma \ref{enalphadeltax}, the paragraph above it, and Lemma \ref{fixedpointestimates}, but will be
   reviewed as needed below.  The bounds depend on $v$ but this dependence won't appear explicitly until variations in $v$ are considered.

 For notational convenience, we set 
        \begin{align}
          \Phi(f) =  \langle f, \Phi\rangle
                           \end{align}
 with $\Phi_t(x) \equiv \left(\phi_t(x), \pi_t(x),r_t\right)^{\intercal}$  regarded as
a column vector.  The equations 
of motion can then be written as
 \begin{align}\label{lineq1}
         d_t\Phi_t(x) = {\mA}_v\Phi_t(x) dt +\sqrt {\bf T}d\omega(t)
         \end{align}
a stochastic differential equation for an infinite-dimensional Ornstein-Uhlenbeck process with solution $\Phi_t$ a stochastic integral (mild solution, \cite{DaPZ}),
    \begin{align}\label{stochasticevol}
                   \Phi_t= \int_0^t \exp((t-s){\mA}_v)\sqrt{\mT}d\omega(s)+\exp {(t\mA_v)}\Phi_0.
                        \end{align}
For $f_{\nvec}$ an eigenvector of $\mA_v^{\dagger}$ with eigenvalue $\lambda_{\nvec}^*$, we have simply that
       \begin{align}\label{stochasticevolfn}
                   \Phi_t(f_{\nvec})= \int_0^t \exp((t-s){\lambda_{\nvec}})\langle f_{\nvec},\sqrt{\mT}d\omega(s)\rangle +\exp {(t\lambda_{\nvec})}\Phi_0(f_{\nvec}).
                        \end{align}

           The process defined by Eq.(\ref{lineq1}) accommodates a stationary Gaussian measure $\mu_v$ depending on $v$. Its 
           {\em mode covariance matrix} $\widehat{C}_v(\mvec, \nvec)$, for each $\mvec, \nvec$ a $4\times4$ matrix,  is given by        
\begin{align}\label{Chat}
     \widehat{C}_v(\mvec, \nvec)\equiv E_{\mu_v}[{P^*_{\mvec}\Phi}P_{\nvec}\Phi]= -\frac{P^{\phantom{*}}_{\nvec} {\bf T}P^{\dagger}_{\mvec}}{\left({\lambda}^*_\mvec +\lambda^{\phantom{*}}_\nvec\right)}.
\end{align}
$\widehat{C}_v$ depends on $v$, through its dependence on the eigenfunctions and eigenvalues, but again
we do not write their dependence until needed in the non-linear perturbation problem.      To see (\ref{Chat}), we have from the equation of motion (\ref{lineq1}) that 
             \begin{equation}
                              d\Phi_t(f_{\nvec}) = \lambda_{\nvec}\Phi_t(f_{\nvec})dt +\langle f_{\nvec},\sqrt{\bf T}d\omega(t)\rangle.
                              \end{equation}
Stationarity of the measure and application of It\^{o}'s lemma gives
\begin{align}
       & d_tE_{\mu_v}[\Phi_t(f^*_{\mvec})\Phi_t(f_{\nvec})]_{\big| t=0} = 0 \nonumber\\
                 &= \left({\lambda}^*_\mvec +\lambda_\nvec\right)E_\mu[{\Phi}(f^*_{\mvec})\Phi(f_{\nvec})] + \langle f_{\nvec}, {\bf T} f_{\mvec}\rangle,
                    \end{align}   
    so that 
    \begin{align}
           E_{\mu_v}[{\Phi}(f^*_{\mvec})\Phi(f_{\nvec})] = -\frac{\langle f_{\nvec}, {\bf T} f_{\mvec}\rangle}{\lambda^*_\mvec +\lambda_\nvec}.
    \end{align}  
    (This relation can also be obtained by computing the covariance of the integral for $\Phi_t(f_{\nvec})$, Eq.(\ref{stochasticevolfn}), in the limit $t\rightarrow\infty$. It
    is evident in this latter approach that $C_v$ is positive definite.)
    Inserting the definition of the projections  (\ref{projectiondef}) completes the derivation of the covariance matrix (\ref{Chat}).   
The denominators in (\ref{Chat}) appears dangerous when $\nvec= \mvec$, $\lambda^*_{\nvec}+\lambda_{\nvec}\sim -|\widehat{\alpha}(n)|^2 /n^2$, or when $\lambda_\mvec$ and $\lambda_\nvec$ are nearly degenerate, ${\lambda}^*_{\mvec}+\lambda_{\nvec}\sim \pm i|\widehat{\alpha}(n)|^2/n$.  
  But in the numerator, only the $r$-components of $f_{\mvec}$ and $f_{\nvec}$ are detected by ${\mT}$ and these
components are $\sim |\widehat{\alpha}(m)| /m$ and $|\widehat{\alpha}(n)|/n$ respectively. See the discussion below
Eq.(\ref{covdefnition}) concerning the field covariance, as well as  Lemma(\ref{conjP}) and Eqs.(\ref{intermsofpi}) in the appendix 
 for a more refined estimate of the eigenvalues.

The field $\phi_t(x)$,  i.e., the first component of $\Phi_t$, can be expanded using the $P_{\nvec}$'s in a resolution of the identity,
    \begin{align}
    \phi_t(x) = \sum_{\nvec}P_{\nvec,\phi(x)}\Phi_t=  \sum_{\nvec}\frac{e_{\nvec,\phi}(x)\Phi_t(f_{\nvec})}{\langle f_\nvec, e_{\nvec}\rangle_{\mathcal{H}}}
                           \end{align}  
($P_{\nvec}\Phi_t$ is a $4$-vector, and, $P_{\nvec,\phi(x)}\Phi_t$ is its $\phi$-component at $x$ and time $t$, $\Phi_t(f_{\nvec})$ given above in (\ref{stochasticevolfn}).)   
With this expansion, we can express the {\em space-time covariance} just for the field $\phi$ using the mode covariance,
\begin{align}\label{covdefnition}
             &C_v(x,y,t)\equiv E_{\mu_v}\left[\phi(x)\phi_t(y)\right]\nonumber\\
             & = \mathop{\sum\sum}_{\mvec,\nvec}\frac{ e^*_{\mvec,\phi}(x) e^{t \lambda_{\nvec}}e^{\phantom{*}}_{\nvec,\phi}(y)}{\langle f_{\mvec},e_{\mvec}\rangle^*_{\mathcal{H}}\langle f_{\nvec},e_{\nvec}\rangle_{\mathcal{H}}}E_{\mu_v}\left[{\Phi}(f^*_{\mvec})\Phi(f^{\phantom{*}}_\nvec)  \right]\nonumber\\
             &=-\mathop{\sum\sum}_{\mvec, \nvec} 
\frac{P^{\phantom{\dagger}}_{\nvec,\phi(y)}TP^{\dagger}_{\mvec,\phi(x)}
}{\lambda^*_{\mvec}+\lambda^{\phantom{*}}_{\nvec}} e^{t\lambda_{\nvec}}.             \end{align}
$C_v(x,y,t)$ is bounded and H\"{o}lder continuous in its variables, with any index $\gamma<1$.   Moreover this continuity implies
by the Kolmogorov continuity  theorem \cite{SV} that $\phi_t(x)$ is a.s. H\"{o}lder continuous with any index $\gamma<1/2$ in both $x$ and $t$  with
 respect to the Gaussian measure $\mu_v$.

Showing the convergence of the series (\ref{covdefnition}) for the covariance and its continuity properties in $x$ and $t$ requires detailed estimates on the eigenfunctions and their eigenvalues which
are established in the appendix (paragraph preceeding
Lemma\ref{enalphadeltax}, Lemmas \ref{enalphadeltax}, \ref{conjP}).  The salient facts are these: (i) There exist constants $\{a_{\nvec}\}$ uniformly bounded in $\nvec$ such that
$P_{\nvec,\phi(x),r}\sim {a_{\nvec} e_{\nvec,\pi}(x)}/{n^2}$, 
with $\|e_{\nvec,\pi}\|_{L^2[0,2\pi]}=1$, ($e_{\mvec,\phi}= e_{\mvec,\pi}/\lambda_{\mvec})$. Moreover, the $e_{\nvec,\pi}$'s are pointwise bounded
and almost Lipschitz uniformly in $\nvec$ in the sense that there exists  positive constants $b, c_{\gamma}$ such that $\sup_x|e_{\nvec,\pi}(x)|\leq b \|e_{\nvec,\pi}\|_{L^2[0,2\pi]}$ and $|e_{\nvec,\pi}(x) -e_{\nvec,\pi}(y)|\leq c_{\gamma}(|n||x-y| +|x-y|^{\gamma})\|e_{\nvec,\pi}\|_{L^2[0,2\pi]}$ for any $\gamma<1$.   (ii) The eigenvalues  behave like $\lambda_{\nvec}\sim \pm in - |\widehat{\alpha}(n)|^2/n^2$ for large $n$ 
 by the non-degeneracy assumption on the $\alpha$'s.  The small denominator issue $\lambda^*_{\mvec} +\lambda^{\phantom{*}}_{\nvec}\sim
-|\widehat{\alpha}(n)|^2/n^2$, $\mvec=
\nvec$ in (\ref{covdefnition}) and the nearly degenerate case $\lambda^*_\mvec+\lambda_{\nvec}\sim \pm i/n$  are  mollified by the $1/n^4$ behavior of the numerator.  (iii) The estimates here are uniform in the potential $v$, provided
$\|v\|_{{\mathcal{C}}}\leq \varepsilon_o \|\widehat{\alpha}\|^2_{\infty}$, for some positive $\varepsilon_o>0$, ($\|v(x)\|_{\mathcal{C}}\equiv \sup_x|v(x)|$). Similar estimates
hold for the time $t$ in the summands of (\ref{covdefnition}).   As example, the diagonal terms of the double sum behave as $n^{-2}$, the near diagonal
terms as $n^{-3}$, and otherwise the terms are $\sim n^{-2}m^{-2}|n-m|^{-1}$ for $n,m$ large, and absolute covergence follows.

 \end{subsection}
 
\begin{subsection}{\bf Semigroup and Generator for $\Phi_t$}
 
  Let $F$ be a functional of the field, e.g., a polynomial in $\Phi(f_{\nvec_1}),...,\Phi(f_{\nvec_k})$, $f_{\nvec_1},..., f_{\nvec_k}$ eigenfunctions of $\mA^{\dagger}$ (these polynomials are dense in the $L^2$-Hilbert space of functionals, with $\mu_v$ as measure).  The semigroup $e^{t{\mathcal{L}}_v}$ associated  with the above process  $\Phi_t$ is defined by  
         \begin{align}\label{semigroup}
         e^{t{\mathcal{L}}_v}F(\Phi) = E_{\Phi}[F(\Phi_t)],
         \end{align}
where the expectation here is with respect to the Brownian motion with  $\Phi_t$ given above (\ref{stochasticevol}) and with $\Phi= \Phi_0$.
       The generator of the semigroup associated with $\Phi_t$ has the 
expression
\begin{align}
    {\mathcal{L}}_vF(\Phi)= \frac{1}{2}\sum_{i=1}^2 \frac{\partial}{\partial r_i} T_i\frac{\partial}{\partial r_i}F(\Phi) + \left\langle \mA_v\Phi, \frac{\delta}{\delta \Phi}F(\Phi)\right\rangle,
\end{align}
             $\frac{\delta}{\delta \Phi} = \left(\frac{\delta}{\delta \phi(x)}, \frac{\delta}{\delta \pi(x)}, \frac{\partial}{\partial r}\right)^{\intercal}$ being the functional
       gradient.  
${\mathcal{L}}_v$ is a well-defined operator acting in $L^2$ of
$\mu_v$. Polynomials in $\Phi$ of degree $m$ are mapped into
polynomials of the same degree by ${\mathcal{L}}_v$ and its
semigroup. There is  ``second quantization'' associated with ${\mathcal{L}}_v$. As simple example, if $f_{\nvec}$ is an eigenfunction of
$\mA_v$ 
with eigenvalue $\lambda^*_{\nvec}$, then $\Phi(f_{\nvec})$ is an eigenfunction of ${\mathcal{L}}_v$ with eigenvalue $\lambda_{\nvec}$; $\Phi(f_{\nvec})$ is
an order $1$ Hermite function of the field.  Given two eigenfunctions  $f_{\mvec}$ and $f_{\nvec}$ of $\mA_v$ , the functional  $F(\Phi)= \Phi(f_{\mvec})\Phi(f_{\nvec})+\langle f_{\mvec},
{\bf T}f^*_{\nvec,r}\rangle/(\lambda_m+\lambda_{\nvec})$, an order $2$ multivariable Hermite function of the field,
 is also an eigenfunction of ${\mathcal{L}}_v$ with eigenvalue $\lambda_{\mvec,\nvec}= \lambda_{\mvec}+\lambda_{\nvec}$. (Note that the denominator $(\lambda_m+\lambda_{\nvec})$ has non-zero real part so it does not vanish.) 
 In general, one can construct higher degree polynomials 
which are also eigenfunctions of ${\mathcal{L}}_v$ of the form: a monomial of degree $k$, $\Phi(f_{\nvec_1}) \Phi(f_{\nvec_2})\cdots\Phi(f_{\nvec_k})$ minus a lower order polynomial, with eigenvalue $\lambda= \lambda_{\nvec_1}+\lambda_{\nvec_2}\cdots+\lambda_{\nvec_k}$. 
All of the eigenfunctions  constructed are of finite $L^2$-norm with respect to $\mu_v$, as a consequence of the boundedness of the mode covariance
and Wick's theorem. They are complete since the $f_n$'s are.  ${\mathcal{L}}_v$ has discrete spectrum.  Clearly  $0$ is an eigenvalue with eigenvector a constant,
but it is also an accumulation point, e.g., $\lambda_{\nvec,\nvec^*} =\lambda_{\nvec} + \lambda^*_{\nvec}$ is an eigenvalue which goes to zero, $n\rightarrow\infty$. That
zero is an accumulation point is the source of small-denominator resonance difficulties.\\

We remark that the adjoint operator  ${\mathcal{L}}^{\dagger}_v$ defined
with respect to $\mu_v$  has the same second derivative part as ${\mathcal{L}}_v$ but 
the first derivative drift part is different.   The dual process is
another Ornstein-Uhlenbeck process. 

 
 \end{subsection}
 \end{section}

\begin{section} {Perturbation by a $\phi^3$ non-linearity}

The goals of this section are to  outline a perturbation expansion, introduce the renormalization issue, and address it  at least to
 second order in a coupling constant $\cc$. The original equations of motion (\ref{equationsofmotion.1}) have the potential $v(x)$ equal to zero.
  But it will be important to keep $v$ in the present calculations. 
 
 \begin{subsection}{Perturbation expansion}
 
 Let  $V= \int dx \phi^3(x)\frac{\delta}{\delta \pi(x)}$ be the perturbation, $\cc$ a coupling constant, and 
let $\nu_{\cc,v}$ be the putative invariant measure  to equations of motion (\ref{equationsofmotion.1}) with  $g(\phi(x))=\cc \phi^3(x)$
and $-v(x)\phi(x)$ included. Let ${\mathcal{L}}_{\cc,v}= {\mathcal{L}}_v -\cc V$ be the associated
(Liouville) semigroup generator.   
Then the formal
 expansion for $\nu_{\cc,v}$, (generating functional)  is given by
\begin{align}\label{EXpansion}
           \int d\nu_{\cc,v}(\Phi)F(\Phi) = \int d\mu_v(\Phi) \sum_{m\geq 0}\cc^m  \left( V\frac{Q_v}{{\mathcal{L}}_v}\right)^m F(\Phi)
  \end{align} 
with  $Q_vF(\Phi)= F(\Phi)-\int d \mu_v F$, $F$ a functional of the field. The series is obtained by making an ultra-violet cutoff so that $0$ is an isolated eigenvalue of ${\mathcal{L}}_v$,  writing then the difference
$\nu_{\cc,v}-\mu_v$ as a  contour integral over a small circle about $0$  with integrand the difference of resolvents, $(z-{\mathcal{L}}_{\cc,v})^{-1} -
(z-{\mathcal{L}}_v)^{-1}$ written via the second resolvent equation  \cite{K} and expanding the difference in a formal Born series. Note that the series for $\nu_{\cc,v}$ has integral $1$, since $\mu_v$ does and the facts that $V$ and $Q_v$ (which appears gratuitously) kill constants; ${\mathcal{L}}_{\cc, v}$ also
has an eigenvalue $0$.  

It is instructive to compute the first order and second order terms in the expansion (\ref{EXpansion}), where the renormalization issue first arises. 
In preparation, we will represent $V{Q_v}{\mathcal{L}}^{-1}_v$ operating on a functional as a time integral,
  \begin{align}
   V\frac{Q_v}{{\mathcal{L}}_v}{F}(\Phi)= -\int_0^{\infty} dt Ve^{t{\mathcal{L}}_v}Q_v{ F}(\Phi),
\end{align}
 ignoring the issue of which $F$'s are in the domain of the operation, (although polynomials in the $\Phi(f_{\nvec})$'s certainly are).
Since $V$ kills constants we can replace the $Q_v$ simply by the identity, and use the semigroup representation
 Eqs.(\ref{semigroup},\ref{stochasticevol}),
    \begin{align}
             V\frac{Q_v}{{\mathcal{L}}_{v}}{ F}(\Phi)= -\int_0^{\infty}dt\, VE_{\Phi}\left[{F}(\Phi_t)  \right],
                         \end{align}
                          the expectation $E[\cdot]$ here with respect to  Brownian motion.  

             We define the function $D_v(x,t; z)$,
           \begin{align}\label{Ddefinition}
              D_v(x,t;z)&\equiv \frac{\delta}{\delta \pi(z)}\phi_t(x)=\frac{\delta}{\delta \pi(z)}\sum_{\nvec}P_{\nvec,\phi(x)}\Phi_t = \sum_{\nvec}\frac{e_{\nvec,\phi}(x)e^{t\lambda_\nvec} f^*_{\nvec,\pi}(z)}{\langle f_{\nvec},e_{\nvec}\rangle}\nonumber\\
                 & = \sum_{\nvec} e^{t\lambda_{\nvec}}P_{\nvec,\phi(x),\pi(z)}
                                        \end{align}  
which will be ubiquitous in the analysis ($P_{\nvec,\phi,\pi}$ is the $\phi,\pi$ entry of the matrix $P_{\nvec}$). Its manipulations are  analogous  to normal ordering
of creation/annihilation operators. The series is only conditionally convergent, $t>0$. 

   For the remainder of the discussion, we confine attention just to the two point function case, $F(\Phi)= \phi(x)\phi(y)$.  For this choice, the leading term in the expansion (\ref{EXpansion}) is         
\begin{align}
      & \cc E_{\mu_v}\!\!\left[V\frac{Q_v}{{\mathcal{L}}_v} F(\Phi)\right]= \cc E_{\mu_v}\!\!\!\left[\mathop{\iint}_{\substack{ t\in [0\infty), \\z\in [0,2\pi]}}\!\!\!\! dt dz \phi^3(z) (\phi_t(x)D_v(y,t;z) +\phi_t(y)D_v(x,t;z)  \right]\nonumber\\
                    & =3\cc \!\!\!\!\mathop{\iint}_{\substack{t\in [0\infty)\\ z\in [0,2\pi]}} dt\,dz \,C_v(z,z,0)\big(C_v(z,x,t)D_v(y,t;z)\! + \!C_v(z,y,t)D_v(x,t;z)\big)
     \end{align}
by Wick's theorem. The expression corresponds to the two diagrams shown:

\begin{figure}[h]
  \begin{center}
  \begin{tikzpicture}[scale=.8]
\Vertex[size=.01 , position=180,label=$z$,style={line width=.5pt},color=white
        ]{2}
\Vertex[x=3,y=-.3,size=.05 ,position=below, label=$x$,style={line width=.5pt},color=white
        ]{1}
\Vertex[x=3,y=.3,size=.1 ,position=above, label=$y$,style={line width=.5pt},color=white
]{0}
\Edge[loopsize=1.5cm,loopposition=270,loopshape=30, lw=.5pt ,label=$C_v$,position={left=1mm}  ](2)(2)
\Edge[bend=-30, lw=.5pt ,label=$C_{v}$ 
](2)(1)
\Edge[bend=30, lw=.5pt ,label=$D_v$ 
](2)(0)

\end{tikzpicture}
\hspace{2cm}
  \begin{tikzpicture}[scale=.8]
\Vertex[size=.01 , position=180,label=$z$,style={line width=.5pt},color=white
        ]{2}
\Vertex[x=3,y=-.3,size=.05 ,position=below, label=$x$,style={line width=.5pt},color=white
        ]{1}
\Vertex[x=3,y=.3,size=.1 ,position=above, label=$y$,style={line width=.5pt},color=white
]{0}
\Edge[loopsize=1.5cm,loopposition=270,loopshape=30, lw=.5pt ,label=$C_v$,position={left=1mm}
      ](2)(2)
\Edge[bend=-30, lw=.5pt ,label=$C_{v}$
](2)(0)
\Edge[bend=30, lw=.5pt ,label=$D_{v}$
](2)(1)

\end{tikzpicture}
\end{center}
\vspace{-.7cm}

\caption{$\cc^1$-diagrams. The upper loops correspond to the $D_v$ functions, the lower loops to the covariances $C_v$. 
 The horizontal direction is time.}
   \end{figure}
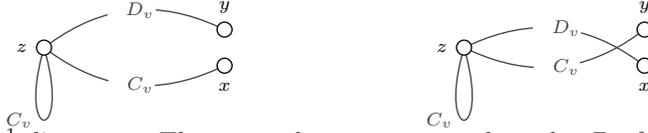

Substituting the mode expansions for $C_v(z,\cdot,t)$ and  $D_v(\cdot,t;z)$, then doing the $t$ integral, but only retaining the resonant contributions  where
an $e^{t\lambda_{\nvec}}$ is integrated against $e^{t\lambda^*_{\nvec}}$ and keeping just the diagonal part of the double series for $C_v$, we obtain a series
        \begin{align}\label{23skidoo} 
           3\cc  \sum_{\nvec} \int \!dzC_v(z,z,0) \frac{P^{\phantom{*}}_{\nvec,\phi(x)} {\bf T}P^{\dagger}_{\nvec,\phi(z)}}{(\lambda_{\nvec}+\lambda^*_{\nvec})^2}P_{\nvec,\phi(y),\pi(z)}
             \end{align}
             and an identical one with $x$ and $y$ interchanged.  The off-diagonal, i.e., non-resonant or nearly resonant, contributions are absolutely summable.   The convergence of this sum (\ref{23skidoo}) is somewhat problematic; in terms of superficial degree of divergence \cite{HZ}, the $P_{\nvec,\phi(x),r}\sim 1/n^2$, $P_{\nvec,\phi(x),\pi(z)}\sim 
          1/n$ and the small denominator $(\lambda_{\nvec} +\lambda^*_{\nvec})^{-2}\sim n^4$.  Thus the $\nvec$ term is of degree $-1$ and the sum of degree $0$. 
          For this diagram however, and with further analysis of the $e_{\nvec,\pi}$'s and $f_{\nvec,\pi}$'s, there is a fortuitous effective cancellation of the  $\nvec$ and $-\nvec$ modes for large $n$, and with this cancellation the sum is finite. The $z$-integral does not enhance the convergence of the sum.
 
 Terms of order $\cc^2$  in Eq. (\ref{EXpansion}), correspond to the associated  diagrams shown, Fig. 2, (and $x$ and $y$ exchanged). For example, the first diagram,  the breaching whale,  corresponds to the integral,
 \begin{figure}[h]
\begin{center}
\begin{tikzpicture}[scale=.4]
\Vertex[size=.1 ,label=$z_1$,position={above}, style={line width=.5pt},color=white
        ]{3}
\Vertex[x=3,size=.1 ,label=$z_2$,position={above}, style={line width=.5pt},color=white
        ]{2}
\Vertex[x=6,y=-.3,size=.1, position=below,label=$x$, style={line width=.5pt},color=white
]{1}
\Vertex[x=6,y=.3,size=.1, position= above,label=$y$, style={line width=.5pt},color=white
]{0}
\Edge[bend=45, lw=.5pt 
      ](3)(0)
\Edge[bend=45, lw=.5pt 
      ](2)(1)
\Edge[bend=-30, lw=.5pt 
      ](3)(2)
\Edge[bend=-50, lw=.5pt 
      ](3)(2)
\Edge[bend=-80, lw=.5pt 
      ](3)(2)
 \end{tikzpicture}
\begin{tikzpicture}[scale=.4]
\Vertex[size=.1 ,label=$z_1$,position={above}, style={line width=.5pt},color=white
        ]{3}
\Vertex[x=3,size=.1 ,label=$z_2$,position={above}, style={line width=.5pt},color=white
        ]{2}
\Vertex[x=6,y=-.3,size=.1, position=below,label=$x$, style={line width=.5pt},color=white
]{1}
\Vertex[x=6,y=.3,size=.1, position= above,label=$y$, style={line width=.5pt},color=white
]{0}
\Edge[bend=45, lw=.5pt 
      ](3)(2)
\Edge[bend=45, lw=.5pt 
      ](2)(0)
\Edge[bend=-30, lw=.5pt 
      ](3)(2)
\Edge[bend=-30, lw=.5pt 
      ](2)(1)
\Edge[loopsize=1.5cm,loopposition=270,loopshape=30, lw=.5pt 
      ](3)(3)
 \end{tikzpicture}
\begin{tikzpicture}[scale=.4]
\Vertex[size=.1 ,label=$z_1$,position={above}, style={line width=.5pt},color=white
        ]{3}
\Vertex[x=3,size=.1 ,label=$z_2$, position={above}, style={line width=.5pt},color=white
        ]{2}
\Vertex[x=6,y=-.3,size=.1, position=below,label=$x$, style={line width=.5pt},color=white
]{1}
\Vertex[x=6,y=.3,size=.1, position= above,label=$y$, style={line width=.5pt},color=white
]{0}
\Edge[bend=45, lw=.5pt 
      ](3)(0)
\Edge[bend=45, lw=.5pt 
      ](2)(1)
\Edge[bend=-30, lw=.5pt 
      ](3)(2)
\Edge[loopsize=1.5cm,loopposition=270,loopshape=30, lw=.5pt 
      ](3)(3)
\Edge[loopsize=1.5cm,loopposition=270,loopshape=30, lw=.5pt 
      ](2)(2)
 \end{tikzpicture}
\begin{tikzpicture}[scale=.4]
\Vertex[size=.1 ,label=$z_1$,position={above}, style={line width=.5pt}, color=white
        ]{3}
\Vertex[x=3,size=.1 ,label=$z_2$,position={above}, style={line width=.5pt},color=white
        ]{2}
\Vertex[x=6,y=-.3,size=.1, position=below,label=$x$, style={line width=.5pt},color=white
]{1}
\Vertex[x=6,y=.3,size=.1, position= above,label=$y$, style={line width=.5pt},color=white
]{0}
\Edge[bend=45, lw=.5pt 
      ](3)(2)
\Edge[bend=45, lw=.5pt 
      ](2)(0)
\Edge[bend=-45, lw=.5pt 
      ](3)(1)
\Edge[loopsize=1.5cm,loopposition=270,loopshape=30, lw=.5pt 
      ](2)(2)
\Edge[loopsize=1.5cm,loopposition=270,loopshape=30, lw=.5pt 
      ](3)(3)
 \end{tikzpicture}
\caption{$\cc^2$-diagrams. Terminal vertices $x$ and $y$ can be interchanged.}

\end{center}

 \end{figure}
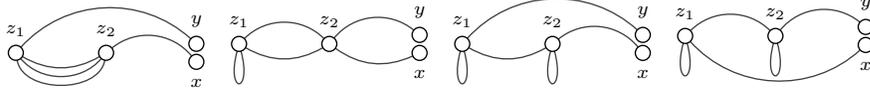
\begin{align}
6\cc^2 \iiiint dt_1  dz_1 dt_2 dz_2 D_v(y,t_1+t_2;z_1)D_v(x,t_2;z_2) C_v^3(z_1, z_2, t_1).
   \end{align} 

 The diagrams encode the rules for a diagram with an arbitrary number of vertices. Each vertex, other than the terminal vertices $x, y$, has valence $3$ and
 is the source of a $D$-loop,  $x$ and $y$ each having 
 one unit of valence.  A $D$-loop consumes $1$ unit of  
 valence only at its right end.  A $C$-loop consumes one unit of
 valence at each of its ends.   An order $\cc^m$-diagram with $m$
 vertices plus the $x,y$ terminals, has $m$ $D$-loops and $m+1$
 $C$-loops. In effect, $D$ has summands of superficial degree of
 divergence $-1$, behaving as $n^{-1}$ and $C$ has summands of degree
 $-2$, behaving as $n^{-2}$ (the diagonal part $\mvec=\nvec$ being of
 highest degree). A  diagram involves an $m +(m+1)= (2m+1)$-fold
 summation over the modes, 
 counting $C$ as a single sum.
 
For a $\cc^m$-diagram there are $m$ time integrals to be performed.  Retaining just the most singular resonant contributions results in $m$ constraints among the modes so the total summation becomes $(m+1)$-fold.  But each time integral, so constrained to the resonant
case introduces a small denominator which can be of degree $+2$.  Thus the superficial degree of a diagram with $m$ vertices plus the terminal $x,y$ ones
is  $D$-degrees+ $C$-degrees +small denominators +summations $= -m -2(m+1)+ 2m+ (m+1) = -1$.

 For the  diagrams in Fig. 2,  $m=2$; there are $2$ $D$-loops
and $3$ $C$-loops, and $2$ time-integrals.   For the breaching whale, the first of the diagrams, the time integrals result in the new small denominators $\lambda^{\phantom{*}}_{\nvec}+\lambda^*_{\nvec}$ and $(\lambda^{\phantom{*}}_{\nvec}\!+\lambda^*_{\nvec_1}\!+\lambda^*_{\nvec_2}\!+\lambda^*_{\nvec_3})$ (the latter with $\lambda$'s with negative real parts, so the sum is no smaller than
$\sim \sum_i |\widehat{\alpha}(n_i)|^2/(n_i)^2$).  Each of these small denominators  contributes a $+2$ to the degree.  The summation is $3= (m+1)$-fold. 
The resulting integral 
\begin{align}
        -6 \cc^2\!\! \mathop{\sum\sum\sum}_{\substack{\nvec_1,\nvec_2, \nvec_3\\ n=n_1+n_2+n_3}} \iint\!\!\prod_{i=1}^3 \frac{P^{\phantom{\dagger}}_{\nvec_i,\phi(z_2)}{\bf T}P^{\dagger}_{\nvec_i,\phi(z_1)}}{(\lambda^{\phantom{*}}_{\nvec_i}+\lambda^*_{\nvec_i})} \frac{P^{\phantom{*}}_{\nvec,\phi(y),\pi(z_1)}P^*_{\nvec,\phi(x),\pi(z_2)} \,\,dz_1 dz_2}{(\lambda^{\phantom{*}}_{\nvec}\!+\lambda^*_{\nvec_1}\!+\lambda^*_{\nvec_2}\!+\lambda^*_{\nvec_3})(\lambda^{\phantom{*}}_{\nvec}\!\!+\lambda^*_{\nvec})}
\end{align}
is finite.\\

The third diagram  in Fig. 2 , (breaching whale with tadpole earrings),  has the integral \begin{align}
    9 \cc^2  {\mathop{\iiiint}}dt_1 dz_1dt_2 dz_2 & C_v(z_1,z_1,0)C_v(z_2,z_2,0) C_v(z_1,z_2,t_1)\times \big. \nonumber\\
           &\big.  D_v(y,t_1+t_2; z_1)D_v(x,t_2;z_2).  
\end{align}
Again, doing the time integrals and retaining just the resonant terms results in the single sum
\begin{align}
  -9\cc^2 \sum_{\nvec} \iint dz_1 dz_2 & C_v(z_1,z_1,0)C_v(z_2,z_2,0) \; \times
  \nonumber\\
 & \frac{P^{\phantom{\dagger}}_{\nvec,\phi(z_2)}\!{\bf T}P^{\dagger}_{\nvec,\phi(z_1)}}{(\lambda^{\phantom{*}}_{\nvec}+\lambda^*_{\nvec})^3} {P^{\phantom{*}}_{\nvec,\phi(y),\pi(z_1)}P^*_{\nvec,\phi(x),\pi(z_2)}}.
\end{align}
Including the superficial degree of each of the tadpoles as $-1$, $C_v(z_i,z_i,0)$ counts
as $-1$,  the over all degree of the diagram is still $-1$.  ( $\int dz C_v(z,z,0)e^*_{\nvec,\pi}(z)f_{\nvec,\pi} (z)$ does not decay, $n\rightarrow\infty$.)    But the remaining summand above
is seen to have degree $0$,  the sum has degree $+1$, and indeed the sum is divergent. The fourth diagram in Fig. 2 has
integral which also diverges and which does not cancel that of the third.
  The remedy is to get rid of diagrams with tadpoles.

\end{subsection}
\begin{subsection}{Simple renormalization} 
 
Diagrams of any order in $\cc$ with these simple tadpoles are eliminated by a renormalization, an appropriate alteration in the measure $\mu_v$ and shift in the interaction.   We begin with the original problem with the linear dynamics having  the potential $v(x)= 0$ and generator ${\mathcal{L}}_o$.    
  The idea is to subtract a  linear drift term  from $ {\mathcal{L}}_o$ and add it in as part of the perturbation, 
                       \begin{align}
                                &{\mathcal{L}}_o\rightarrow {\mathcal{L}}_{v}\equiv{\mathcal{L}}_o-\int dx v(x) \phi(x)\frac{\delta}{\delta \pi(x)}, \text{ and} \\
                                      & -\cc\int dx\phi^3(x)\frac{\delta}{\delta \pi(x)} \rightarrow -\int dx (\cc \phi^3(x) -v(x)\phi(x))\frac{\delta}{\delta \pi(x)}
                                       \end{align} 
in such a manner that $\cc \phi^3(x)-v(x)\phi(x)$ is, {\em for all} $x$, a third order Wick polynomial with respect to $\mu_v$,  the  Gaussian measure  invariant
under the evolution generated by ${\mathcal{L}}_v$.  This is the case if
                          \begin{align}\label{tadpolekiller}
                      v(x) = 3\cc E_{\mu_v}[\phi^2(x)].
               \end{align}
               With this choice for $v$, the additional term in the interaction will cancel the self-loops.
This {\em implicit} equation for $v$ has a unique solution for small coupling constant $\cc$ by the contraction mapping theorem \cite{RS}.  In outline, and recalling
the expansion for the covariance (\ref{covdefnition}), we consider the mapping 
        \begin{align}
         v\rightarrow \tilde{v}_v\equiv C_v(x,x,0)= -\mathop{\sum\sum}_{\mvec, \nvec}\frac{ P^{\phantom{\dagger}}_{\mvec,\phi(x)}(v),{\bf T}
 P^{\dagger}_{\nvec,\phi(x)}(v)}{\lambda_{\mvec}(v)+\lambda_{\nvec}^*(v)}.
         \end{align}
   Lemma(\ref{fixedpointestimates}) of the appendix 
provides estimates on the $v$-dependence of the projections and eigenvalues which imply for some $\varepsilon>0$ and $v$'s in the ball $B= \{v: \|v\|_{\mathcal{C}}\leq \varepsilon \|\widehat{\alpha}\|^2_{\infty}\}$,  the map       
is Lipschitz in $v$, $\|\tilde{v}_{v_2}-\tilde{v}_{v_1}\|_{\mathcal{C}}\leq C\|v_2-v_1\|_{\mathcal{C}}$, for some constant $C$. (We know that $\tilde v$ is a continuous
function of $x$.)   Note particularly that the second line of (ii) of the Lemma shows that the small denominator, $\lambda_{\nvec}+\lambda_{\nvec}^*$  moves
in a manageable way as a function of $v$. Thus for small enough coupling constant $\cc$, $v\rightarrow 3\cc\tilde{v}_v$ lands in the
ball $B$ and is a strict contraction, hence the existence of a fixed point.  \\

The above remarks do not exhaust the problem of tadpoles. For example, one can have an order $\cc^2$-\textit{composite} tadpole, as shown in Fig. 3, requiring
higher order corrections to $v$ and solving a new fixed point problem of higher order in $\cc$.    

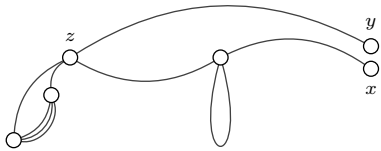
\begin{figure}
\begin{center}
    
    \begin{tikzpicture}
\Vertex[size=.01 , x=-.75,y=-1.1,color=white,style={line width=.5pt}
]{1}
\Vertex[size=.01 , x=-.25,y=-.5, color=white,style={line width=.5pt}
]{2}
\Vertex[size=.01 , x=0,y=0, color=white,style={line width=.5pt}, label=$z$, position=90
]{3}
\Vertex[size=.01 , x=2,y=0, color=white,style={line width=.5pt}
]{4}
\Vertex[size=.01 , x=4,y=-.15, color=white,style={line width=.5pt}, label=$x$,position=270
]{5}
\Vertex[size=.01 , x=4,y=.15, color=white,style={line width=.5pt}, position=90,label=$y$
]{6}
\Edge[bend=-30, lw=.5pt 
](1)(2)
\Edge[bend=-45, lw=.5pt 
](1)(2)
\Edge[bend=-60, lw=.5pt 
](1)(2)
\Edge[bend=30, lw=.5pt 
](1)(3)
\Edge[bend=30, lw=.5pt 
](2)(3)
\Edge[bend=-30, lw=.5pt 
](3)(4)
\Edge[loopsize=1.5cm,loopposition=270,loopshape=30, lw=.5pt 
](4)(4)
\Edge[bend=30, lw=.5pt 
](3)(6)
\Edge[bend=30, lw=.5pt 
](4)(5)
      \end{tikzpicture}
      \caption{$\cc^4$-diagram with a composite and a simple tadpole.}

\end{center}

    \end{figure}
  
  \end{subsection}

\end{section}
\begin{section}{Concluding remarks}

The authors have investigated selected higher order integrals (order $\cc^6$ and higher), all of which are seen to be finite.  The 
rank one Brascamp-Lieb inequalities are of utility in estimating these integrals \cite{BL,BL2,BCCT}.
  But a general scheme for showing integrals of arbitrary order are
  finite, that the perturbation expansion is renormalizable, remains
  open.


It is interesting to remark that in the case of equilibrium $T= T_1= T_2$, no renormalization is
 necessary. In this case, $\mu_v$ is a Gaussian Gibbs measure, with the
 momentum $\pi$  uncorrelated with $\phi$ or $r$.  For the
 $\cc^1$-term in the generating  functional Eq. (\ref{EXpansion}), we
 have, letting $V$ operate instead on $\mu_v$,  
 \begin{align}
      \cc &\int d\mu_v(\Phi) \left( V\frac{Q_v}{\mathcal{L}}_v\right) F(\Phi)
      = \frac{\cc}{T} \int d\mu_v(\Phi) \left(\int\phi^3(x)\pi(x) dx\right) \frac{Q_v}{{\mathcal{L}}_v}F(\Phi)\nonumber\\
       &= -\frac{\cc}{T}\left\langle {\mathcal{L}}^{\dagger}_v\left(\frac{1}{4}\int\phi^4(x)dx- C\right),\frac{Q_v}{{\mathcal{L}}_v}F(\Phi)\right\rangle_{L^2(\mu_v)}\nonumber\\
       &=
         -\frac{\cc}{T}\left\langle \left(\frac{1}{4}\int\phi^4(x)dx- C\right), F(\Phi)\right\rangle_{L^2(\mu_v)}
    \end{align}
 the constant $C=E_{\mu_v}[\frac{1}{4}\int\phi^4(x)dx]$ chosen so that
 the expression vanishes if $F$ is a constant.  A similar calculation
 holds for order $m>1$,   resulting
 indeed in the $\frac{1}{m!}\left(\frac{\cc}{4T}\int \phi^4(x)
 dx\right)^m$ contribution, with the whole series formally summing to
 a constant times 
 $\exp \left(-\frac{\cc}{4T}\int \phi^4(x) dx\right)$ as it should. 
${\mathcal{L}}^{\dagger}_v$, the adjoint operator to ${\mathcal{L}}_v$ with
respect to $\mu_v$, is simply ${\mathcal{L}}_v$ but with opposite signs for its
$\phi $ and $\pi $ drift terms. However, it remains to investigate the
 diagram integrals at equilibrium $T_{1}=T_{2}$.

The Gaussian measure $\mu_v$ here has a fairly explicit covariance, and it should be possible to identify 
non-negative Radon-Nikodym factors, i.e., functionals of $\phi$ and $\pi$ and $r$, integrable with respect to $\mu$, to define new
measures on the space of fields. Then the issue would  be to identify physically compelling processes globally
defined in time, modeling non-equilibrium heat flow for which the modified measures are stationary.   


   It remains to investigate the expected current flow in perturbation theory averaged over $[0,2\pi]$, i.e., 
         \begin{align}
         \frac{1}{2\pi} E_{\nu_{\cc,v}}&\left[ \int_0^{2\pi}\!\!\!\! dx\, \pi(x) \partial_x\phi(x)\right] = \nonumber\\
         &  \frac{1}{2\pi}\mathop{\sum\sum}_{\mvec,\nvec}\int_0^{2\pi}\!\!\!\!\! dx \frac{ e^*_{\mvec,\pi}(x) \partial_xe^{\phantom{*}}_{\nvec,\phi}(x)}{\langle f_{\mvec},e_{\mvec}\rangle^*_{\mathcal{H}}\langle f_{\nvec},e_{\nvec}\rangle_{\mathcal{H}}}E_{\nu_{\cc,v}}\left[{\Phi}(f^*_{\mvec})\Phi(f^{\phantom{*}}_\nvec)  \right].
  \end{align}  
  The integral of $\pi$ against $\partial_x\phi_x$ is singular and requires refinements of the estimates  of lemma (\ref{enalphadeltax}) in the appendix, particularly on the real and imaginary parts of $\partial_xe_{\nvec,\phi}(x)$.
The scaling of the ring $[0,2\pi]$ to one of arbitrary size $[0,L]$
should be straightforward. We will address the issue of current flow elsewhere.

\end{section}



\subsection*{Acknowledgment}
We thank D.C.~Brydges for useful conversations and a
reviewer for helpful comments.


\begin{section}{Appendix}
 
Using the eigenvalue equation $\mA_ve_{\nvec}=\lambda_{\nvec}e_{\nvec}$ for the vector $e_{\nvec}$, one can solve for the components  $e_{\nvec,\phi}, e_{\nvec,r}$, just  in terms of $e_{\nvec,\pi}$,  and similarly for $f_{\nvec}$,
\begin{align}\label{intermsofpi}
&e_{\nvec}= \left(\frac{e_{\nvec,\pi}}{\lambda_{\nvec}},e_{\nvec,\pi}, \frac{\langle\alpha,e_{\nvec.\pi}\rangle}{(\lambda_{\nvec}+1)}\right)^{\intercal},\nonumber\\
&f_{\nvec} = \left(\frac{(\partial_x^2-1-v)f_{\nvec,\pi}}{\lambda_{\nvec}^*}, f_{\nvec,\pi},-\frac{\langle f_{\nvec,\pi},\alpha\rangle^{*}}{(\lambda^*_{\nvec,\pi}+1)} \right)^{\intercal}.
\end{align}
 The eigenvalue equation for $\mA_v$ is equivalent to a  Schr\"{o}dinger-like eigenvalue equation, with the eigenvalue appearing implicitly 
      \begin{align} \label{Schrod1}
                          (\partial_x^2-1 - v(x))e_{\nvec,\pi} -\frac{\lambda_{\nvec}}{\lambda_{\nvec}+1} \alpha(x)\cdot\langle\alpha,e_{\nvec,\pi}\rangle
                           = \lambda_{\nvec}^2 e_{\nvec,\pi}.                  
\end{align}
The $f_{\nvec,\pi}$ components for the adjoint problem $\mA^{\dagger}_vf_{\nvec}=\lambda^*_\nvec f_{\nvec}$ are solutions to the complex conjugate
of the above, so $f_{\nvec,\pi}= e_{\nvec,\pi}^*$.
(Note that  there are also two eigenvalues  of $\mA_v$ close to $-1$; their  corresponding eigenvectors have most of their weight
on their $r$-components, with small $\phi$ and $\pi$ components and do not  play an important role in our analysis.)

In order to obtain qualitative information on $e_{\nvec,\pi}$, we first examine eigenfunctions for 
      \begin{align}
         h_{\varepsilon} =       \partial_x^2-1 - v -\varepsilon |\alpha\rangle\cdot\langle\alpha| 
         \end{align}
          acting  in $L^2[0,2\pi]$ with periodic boundary conditions; $v= v(x)$ is real and continuous with bound $\|v\|_{\mathcal{C}}\equiv \sup_x|v(x)|\leq \varepsilon_o\|\widehat{\alpha}\|^2_{\infty}$, $\varepsilon_o$ prescribed below. The $\varepsilon$ is complex, of modulus less than $1$, and ultimately just equal
 to $\lambda_{\nvec}/(\lambda_{\nvec}+1)$.  The $\alpha$'s are as assumed in the text.
 This analysis of $h_{\varepsilon}$ largely reviews  that in
 \cite{T1}, but also provides estimates with $v$ needed for the
 perturbation expansion.

   Let $\{\psi_n\}$ be the $L^2$-normalized eigenfunctions for $h_{\epsilon}$,   $h_{\varepsilon}\psi_n=\lambda_n^2\psi_n$. (The overall sign of
   $\lambda_{\nvec}^2$ on the right side is jarring but in keeping with Eq.(\ref{Schrod1}).) We assume that
   $\|\widehat{\alpha}\|^2_{\infty}$ is small so that $\lambda_{n}^2$ is close to $-(n^2+1)$, the eigenvalue for $\partial_x^2-1$, $|\lambda_n^2+(n^2+1)|\sim\|\widehat{\alpha}\|^2_{\infty} $, uniformly in $n$.  To see this uniformity, suppose first that $v=0$. Then
         the two  components $\langle \alpha_i,\psi_{n}\rangle,\,\,i=1,2$, satisfy
                   \begin{align}
                              \langle \alpha,\psi_{n}\rangle= \varepsilon\langle\alpha,(\partial_x^2-1 -\lambda_{n}^2)^{-1}\alpha\rangle\cdot\langle \alpha,\psi_{n}\rangle.
                                                   \end{align}
 The eigenvalues $\{\lambda^2_{\nvec}\}$ are zeros of a $2\times2$ determinant function of $\lambda^2$. 
  Writing the resolvent $(\partial_x^2-1-\lambda^2_n)^{-1}$ in its spectral representation, pulling out the $n$ and $-n$ terms, and setting $\Delta\lambda^2_n=(-(n^2+1))-\lambda^2_{n}$ as the shift in eigenvalue from that of $\partial_x^2-1$, we have that
 \begin{align}
      \Delta \lambda^2_n\langle \alpha, \psi_n\rangle = \varepsilon\left(\langle\alpha,P^o_{n,-n}\alpha\rangle + \Delta \lambda^2_n\left\langle \alpha,\frac{Q^o_{n,-n}}{(\partial_x^2-1-\lambda_n^2)}\alpha\right\rangle\right)\cdot \langle \alpha,\psi_{n}\rangle
 \end{align} 
 with $P^o_{n,-n}$ projection onto the subspace spanned by $e^{inx}$ and $e^{-inx}$, $Q^o_{n,-n}= {\mathds 1}-P^o_{n,-n}$. 
The $2\times2$-matrix in bra-ket  inside the  parentheses is ${\mathcal{O}}(n^{-\gamma})\|\widehat{\alpha}\|^2_{\infty}$, $\gamma<1$ (with constant depending on $\gamma$),  by the Estimate  \ref{Estimatei-ii} below, so that to leading order,
$\Delta\lambda_n^2$ is just an eigenvalue of  $\varepsilon\langle \alpha, P^o_{n,-n}\alpha\rangle$ and itself 
of size $ |\varepsilon| |\widehat{\alpha}(n)|^2$ with corrections $\sim \Delta\lambda_n^2 n^{-\gamma}$.  By  the non-degenerate condition on the $\alpha$'s,  the two eigenvalues of $\langle \alpha, P^o_{n,-n}\alpha\rangle$ split to the
     same size as their shift, i.e. there exists a positive constant $c_1$ independent of $n$ such that
         \begin{align}\label{eigenshift}
          |\Delta\lambda_n^2-\Delta\lambda_{-n}^2|\geq c_o|\varepsilon| |\widehat{\alpha}(n)|^2.
     \end{align}
    If an additional term $v$  suitably  small is added to $\partial_x^2-1$,  $\|v\|_{\mathcal{C}} \leq |\varepsilon_o |\widehat{\alpha}(n)|^2$, for some
    $\varepsilon_o>0$, Ineq.(\ref{eigenshift}) still holds, but with a  reduced positive constant $c_o$.
 Again, the $\psi_n$'s will be identified with the $e_{\nvec,\pi}$-components. There will
be {\em four}  eigenvalues for $\mA_v$  of "level'' $ n$, $n\neq 0$, two if $n=0$ and two near $-1$. In the case $n\neq 0$,  the eigenvalues are  $\pm i(\sqrt{n^2+1} +\frac{1}{2}\Delta\lambda_{\pm n}^2/\sqrt{n^2+1})$. For $\varepsilon= \lambda_{\nvec}/(\lambda_{\nvec} +1)$ which has an ${\mathcal{O}}(1/n)$ imaginary part, $\lambda_{\nvec}$  acquires a small negative real part $\sim -|\widehat{\alpha}(n)|^2/n^2$.

In the following, we are assuming that the eigenvalues are as above, that there are pairs of eigenvalues $\lambda_n^2$ and $\lambda_{-n}^2$ close
to $-(n^2+1)$ and separated by $\sim \|\widehat{\alpha}\|^2_\infty$. 
  
 \begin{lemma}\label{enalphadeltax} Assume that the eigenfunctions $\{\psi_n\}_{n\in {\mathds Z}}$ are $L^2$-normalized.  Then
     the sequences  $ \{\langle\alpha,\psi_n\rangle\}_{n\in {\mathds Z}}$ and 
      $\{\sup_x |e_n(x)|\}_{n\in {\mathds Z}}$ are uniformly bounded in $n$.  Moreover, there exists a positive constant $c$ and a positive $c_\gamma$ depending on $\gamma<1$ but independent of $n$ such that
       \begin{align}\label{enalphadeltaxeqdelta}
            |\psi_n(x)-\psi_n(y)|\leq c\left(|n||x-y| + c_\gamma|x-y|^\gamma\right)      \end{align}
       for each $n$.        
               There exists a positive constant $c'$ such that
       if  $|\Im \varepsilon|\leq 
       c'\Re \varepsilon$ with $\Re\varepsilon\neq 0$ (and possibly redefining $\psi_n^*$ by an overall phase factor), then
           \begin{align}\label{enalphadeltax.eq2}
           \|\psi_n- \psi_n^*\|_{L^2}\leq c'|\Im\varepsilon|\left(\frac{1}{|\Re\varepsilon|} +\frac{1}{|n|+1}    \right). 
            \end{align}
  These estimates are also uniform in $\varepsilon$ provided $1/2\leq |\varepsilon|\leq 1$  and in the potential $v$,  provided $\|v\|_{\mathcal{C}}\leq \varepsilon_o\| \widehat{\alpha}\|^2_{\infty}$ for some positive $\varepsilon_o$.
   \end{lemma}
   
 \noindent {\em Remark}: The inequality (\ref{enalphadeltax.eq2}) implies as well, at least for small enough $\Im \varepsilon/\Re\varepsilon$ that a similar operator bound holds for
 the difference of projections,
  \begin{align}\label{projectiondifference}
 \left \|\frac{|\psi_n\rangle\langle \psi^*_n|}{\langle \psi^*_n,\psi_n\rangle} - \frac{|\psi_n\rangle\langle \psi_n|}{\langle \psi_n,\psi_n\rangle}\right\|\leq c'|\Im\varepsilon|\left(\frac{1}{|\Re\varepsilon|} +\frac{1}{|n|+1}    \right).
 \end{align}
 
\begin{proof} Let  $P^o_{n,-n}$  be as above and $P_{n,-n}$ projection onto the subspace spanned by $\psi_n$ and $\psi_{-n}$ corresponding to the nearly degenerate eigenvalues $\lambda^2_{n}, \lambda_{-n}^2$.  We assume $\psi_n$ to be normalized and write it as $\psi_n =\xi^o_n +(P_{n,-n} -P^o_{n,-n})\psi_{n}$ with $\xi^o_n= P^o_{n,-n}\psi_n$.  Note that $\|\xi^o_n\|_{L^2}\leq \|\psi_n\|_{L^2}=1$.   The difference of projections  is then expressed as a contour integral of resolvents using the second resolvent identity \cite{K}
\begin{align}\label{enanalysis1}
     \psi_{n} &= \xi^o_n - \frac{1}{2\pi i}\oint_{\Gamma_n} dz (\partial_x^2-1-z)^{-1}(v+\varepsilon |\alpha\rangle\cdot \langle\alpha|)(h_{\varepsilon}-z)^{-1} \psi_n \nonumber
     \\
      & =  \xi^o_n -\!\!\!\!\!\sum_{\{m: \,\,m\neq \pm n\}}\!\!\!\! \frac{\langle \psi^o_m,v\, \psi_n\rangle + \varepsilon \langle \psi^o_m,\alpha\rangle\cdot \langle\alpha,\psi_n\rangle}{m^2+1+\lambda_n^2}\,\psi^o_{m}, 
    \end{align} 
$\{e^o_m\}_{m\in {\mathds Z}}$ being the eigenfunctions of $\partial_x^2-1$. Here $\Gamma_n$ is a loop about  $\lambda_n^2$, its nearby companion $\lambda_{-n}^2$, and $-(n^2+1)$ (which again is doubly degenerate for $ \partial_x^2 -1$, $n\neq 0$), and enclosing no other eigenvalues of $h_\varepsilon$. These eigenvalues are isolated by a distance $\sim |n|$ from other eigenvalues.
   
Multiplying through the representation (\ref{enanalysis1}) by $\langle \alpha|$ gives an implicit estimate for  $\langle\alpha,\psi_n\rangle$
 \begin{align}\label{alpha e}
     |\langle \alpha,\psi_n\rangle| \leq |\widehat{\alpha}(n)| +  \!\!\!\!\!\sum_{\{m: \,\,m\neq \pm n\}}\!\!\!\! \frac{\|\widehat{\alpha}\|_\infty \|v\|_{\mathcal{C}}+|\varepsilon|\|\widehat{\alpha}\|^2_\infty |\langle \alpha,\psi_n\rangle|}{\left|m^2+1+\lambda_n^2\right|}
 \end{align}  
 The coefficient sum in front of $ |\langle\alpha,\psi_n\rangle|$ on the rhs goes as $\sim c_\gamma |n|^{-\gamma}$for any $\gamma<1$ for large $n$ by application of (i) in the Estimate(\ref{Estimatei-ii}) below, since the $\{\widehat{\alpha}(m)\}$ are uniformly bounded and the quotients $\{|m^2-n^2|/|m^2+1+\lambda_n^2|, |m|\neq |n|\}_{m,n}$ are uniformly bounded in $m,n$. This coefficient sum can be made less than $1$ in magnitude for large enough $n$.  Thus the implicit bound Ineq.(\ref{alpha e}) can be resolved
 showing that $\{\langle \alpha,\psi_n\rangle\}_{n\in {\mathds Z}}$ is uniformly bounded.   Moreover, $|\langle \alpha,\psi_n\rangle-\langle\alpha,\xi^o_n\rangle|$
 also goes to zero as $|n|^{-\gamma}$, $n\rightarrow\infty$. (We are assuming that the small $n$ terms are finite.)

 The same argument again using Eq.(\ref{enanalysis1}) shows that the $\psi_n$'s,  are uniformly bounded in $x$, since the $\psi^o_n= e^{inx}/\sqrt{2\pi}$'s and $\xi^o_n$'s are.  Note that $\|\psi_n-\xi^o_n\|_{\mathcal{C}}$ also goes to zero as $n^{-\gamma}$, $n\rightarrow\infty$.   This completes the argument for the first assertion of
 the lemma.

 To show Ineq.(\ref{enalphadeltaxeqdelta}), we note first that $|\psi^o_m(x)- \psi ^o_m(y)|\leq \min\{|m| |x-y|,1\}$ as well as $|\xi^o_n(x)-\xi^o_n(y)|\leq\min\{|n||x-y|,1\}$.
  Eq.(\ref{enanalysis1}) gives the inequality
\begin{align}
|\psi_n(x)-\psi_n(y)|\leq |n||x-y| +K_\varepsilon \!\!\!\!\!\!\!\mathop{\sum}_{\{m: \,\,m\neq \pm n\}}\!\!\!\! \frac{\min\{|m||x-y|, 1\}}{\left|m^2+1+\lambda_n^2\right|},
\end{align}
where $ K_\varepsilon =  \|\widehat{\alpha}\|_\infty \|v\|_{\mathcal{C}}+|\varepsilon|\|\widehat{\alpha}\|^2_\infty | \sup_n|\langle\alpha,\psi_n\rangle|$. 
  We use the uniform boundedness of  the quotients $|m^2-n^2|/|m^2+1+\lambda_n^2|$ and (ii) in the
Estimate(\ref{Estimatei-ii}), with $\eta=|x-y|$ to establish Ineq.(\ref{enalphadeltaxeqdelta}) of the lemma. 

The last assertion of the lemma follows by showing that $\psi_n$ and $\psi^*_n$ are small perturbations of the same eigenfunction for the self-adjoint operator
$h_{\Re\varepsilon}$ where the imaginary part of $\varepsilon$ has been set to zero.  For {\it temporary}
convenience of notation, let  $\{\psi^o_n,\lambda_{o,n}^{2}\}$ be the normalized eigenvectors and eigenvalues of $h_{\Re\varepsilon}$. Then as in Eq.(\ref{enanalysis1}),
\begin{align}
          \psi_n= P^o_n \psi_n  -i\Im\varepsilon\!\!\!\!\sum_{\{m: \,\,m\neq  n\}}\!\!\!\! \frac{ \langle \psi^o_m,\alpha\rangle\cdot \langle\alpha,\psi_n\rangle}{-\lambda_{o,m}^{2}+\lambda_n^2}\,\psi^o_{m}, 
\end{align}
$P^o_n$ projection onto the one-dimensional subspace spanned
by $\psi^o_n$.   ($\Gamma_n$ in this situation is a contour just about $\lambda_{o,n}^{2}$ and $\lambda_n^2$; we have a simple perturbation.)
The near resonant $m= -n$ summand  has a denominator at least $c_o\Re\varepsilon|\widehat{\alpha}(n)|^2$ as a consequence of the non-degeneracy assumption,
see Ineq.(\ref{eigenshift}),
 and so this term is bounded in  $L^2$-norm by $|\Im\varepsilon|/c_o|\Re\varepsilon|$ 
uniformly in $n$.   The rest of the series is bounded in an $L^2$-sense by $|c_1\Im\varepsilon|/|n|$ for a suitable constant $c_1$ from the orthogonality of the
$\psi_m$'s.  Thus 
\begin{align}
         \|\psi_n- P^o_n\psi_n\|_{L^2}\leq |\Im\varepsilon| \left( \frac{1}{c_o|\Re\varepsilon|}+ \frac{c_1}{|n|}\right).
\end{align}

The function $\psi^*_n$ satisfies precisely the same estimates;   $\psi_n$ and $\psi^*_n$ can still be such that $P^o_n\psi_n$ and $P^o_n\psi^*_n$ differ by 
a phase factor, but one can simply choose a phase factor $e^{i\theta_n}$ so that 
\begin{align}
\|\psi_n-e^{i\theta_n}\psi_n^*\|_2\leq  2 |\Im\varepsilon | \left( \frac{1}{c_o|\Re\varepsilon|}+ \frac{c_1}{|n|}\right).
\end{align}  
All of these estimates are uniform in $\varepsilon$ provided
$|\varepsilon| $ is bounded away from zero so that $\lambda^2_{\nvec}$ and $\lambda_{-\nvec}^2$,
which depend on $\varepsilon$, truly separate by $\sim\|\widehat{\alpha}(n)\|^2$.  The estimates are uniform in $v$ for $\|v\|_{\mathcal{C}}$ much smaller in norm than
this same eigenvalue separation.
\end{proof}

\begin{estimate}\label{Estimatei-ii} \textit{There exists a finite  positive constant $c_\gamma,\,\gamma<1$ independent of $n$, such that
for $\eta\geq 0$ and for any integer $n$, }
                 \begin{align}
                ( i)\,\,\, \mathop{\sum}_{m: \,\,
                 m\neq n} \frac{1}{|m^2-n^2|}\leq  \frac{c_\gamma}{1+ |n|^{\gamma}},
               \end{align}
                            \begin{align}
                 (ii)\,\,\,  \mathop{\sum}_{m: \,\, m\neq n} \frac{\min\{|m| \eta, 1\}}{|m^2-n^2|}\leq  c_\gamma \eta^{\gamma}.
                 \end{align}
\end{estimate}

\begin{proof}  It suffices to prove the estimates for $n\geq 0$  and the sums running over $m\geq 0$. \\
 (i). The estimate is clear for $n=0$.  For $n> 0$, we use H\"{o}lder's inequality, 
\begin{align}
         \sum_{m: \,m\neq n}\frac{1}{|m-n|} \frac{1}{|m+n|}\leq \big\|\frac{\chi_{m\neq n}}{(m-n)}\big\|_{\ell^p} \big\| \frac{1}{(m+n)}\big\|_{\ell^{p/(p-1)}}
            \leq \frac{c_p}{ n^{1/p}}
\end{align}
with $c_p$ finite for any $p>1$, thus for $\gamma\equiv 1/p<1$.    \\

(ii) The sum is bounded by
\begin{align}
   \sum_{m: \,m\neq n, 1\leq m\leq 1/\eta} \frac{\eta}{|m-n|} \cdot 1 +  \sum_{m: \,m\neq n, \,1/\eta< m} \frac{1}{|m-n|} \frac{1}{ (m+n)}\nonumber\\
  \leq         \big\|\frac{\chi_{m\neq n}}{(m-n)}\big\|_{\ell^{p/(p-1)}}\left(\eta \big\| \chi_{m\leq 1/\eta}\big\|_{\ell^p} + 
  \big\|\frac{\chi_{m>1/\eta}}{(m+n)}\big\|_{\ell^{p}}\right)
  \leq c_p \,\eta^{(1-1/p)}  
    \end{align}
for any $p>1$ or here, with $\gamma\equiv 1-1/p<1$.  \end{proof}

The projections  $\{P_{\nvec} = \frac{ |e_{\nvec}\rangle  \langle f_{\nvec}|}{\langle e_{\nvec},f_{\nvec}\rangle}\}$,
suitably conjugated, are uniformly bounded  for $n\rightarrow\infty$.  
  Let ${\Lambda}$ be defined by $\Lambda e = \left( (-\partial_x^2+1)^{1/2} e_{\phi}, e_{\pi},e_{r}\right)^\intercal$, putting the $\phi$ and $\pi$ components of
  $e_{\nvec}$ and $f_{\nvec}$ on a comparable $L^2$-norm  footing for large $n$.  
   
 \begin{lemma}\label{conjP}  There is a positive constant $c<1$ independent of $\nvec$ such that 
     \begin{align}\label{Pboundseq7}  c \| {\Lambda}^{-1}f_{\nvec}\|_{\mathcal{H}}\| \Lambda e_{\nvec}\|_{\mathcal{H}}       \leq | \langle f_{\nvec}, e_{\nvec}\rangle_{\mathcal{H}}| \leq \|{ \Lambda}^{-1}f_{\nvec}\|_{\mathcal{H}}\| \Lambda e_{\nvec}\|_{\mathcal{H}}.
     \end{align}
 The conjugated projections $\{\Lambda P_{\nvec}\Lambda^{-1}\}$ are uniformly bounded.
     \end{lemma}
 \begin{proof}  Referring to Eq.(\ref{intermsofpi}) where the components of $e_{\nvec}$ and $f_{\nvec}$ are written in terms of their $\pi$-components, we have
 that the maps $f_{\nvec,\pi}\rightarrow (-\partial_x^2+1)^{-1/2}f_{\nvec,\phi}$ and $e_{\nvec,\pi}\rightarrow (-\partial_x^2+1)^{1/2} e_{\nvec,\phi}$ are uniformly bounded
 acting in $L^2[0,2\pi]$, as are the maps $f_{\nvec,\pi}\rightarrow f_{\nvec,r}$ and $e_{\nvec,\pi}\rightarrow e_{\nvec,r}$, so that
 the left side of (\ref{Pboundseq7} ) is bounded above  by a constant times $\|f_{\nvec,\pi}\|_{L^2}\|e_{\nvec,\pi}\|_{L^2}$.  

 The middle term of (\ref{Pboundseq7} ) can also be written  in terms of the $\pi$-components,   
 \begin{align}\label{innerprodpi} 
        \langle f_{\nvec},e_{\nvec}\rangle_{\mathcal{H}}  &= 2 \int f_{\nvec,\pi}^*(x) e_{\nvec,\pi}(x)\, dx + \frac{1}{\lambda_{\nvec}(1+\lambda_{\nvec})^2}\langle\alpha,f_{\nvec,\pi}\rangle^*\cdot \langle\alpha, e_{\nvec,\pi}\rangle\nonumber\\
                      &= 2\langle f_{\nvec,\pi},e_{\nvec,\pi}\rangle_{L^2} + {\mathcal{O}}(n^{-3})\|f_{\nvec,\pi}\|_{L^2}\|e_{\nvec,\pi}\|_{L^2}\nonumber\\
                      &= 2e^{i\theta_n}\|f_{\nvec,\pi}\|_{L^2} \|e_{\nvec,\pi}\|_{L^2} + {\mathcal{O}}(n^{-3}) \|f_{\nvec,\pi}\|_{L^2} \|e_{\nvec,\pi}\|_{L^2}.
 \end{align}
 Here we have used  the last inequality (\ref{enalphadeltax.eq2}) of Lemma(\ref{enalphadeltax}) with the role of
 $\varepsilon$ played by $\varepsilon= \frac{\lambda_{\nvec}}{\lambda_{\nvec}+1}$ with $|\Im \varepsilon/\Re\varepsilon|\sim 1/n$.
Thus we have a lower bound on the middle term comparable to the left side of (\ref{Pboundseq7} ).

 The second inequality of (\ref{Pboundseq7}) is  an application of
 Cauchy--Schwarz.
 \end{proof}

The following lemma provides the estimates needed for showing  that the mapping $v\rightarrow E_{\mu_v}[\phi^2(x)]$ is Lipschitz continuous in $v$. 
We write 
                      \begin{align}
                          P_{\mvec,\phi(x), r}(v)= -\frac{e_{\mvec,\pi}(x) \langle f_{\mvec,\pi},\alpha\rangle}{\lambda_{\mvec}(\lambda_{\mvec}+1) \langle f_{\mvec},e_{\mvec}\rangle_{\mathcal{H}}}
                                          \end{align}
as the $\phi(x),r$ entry of the matrix $P_{\mvec}$ and an analogous expression for $P_{\mvec,\phi,\pi}$.  The eigenfunctions and eigenvalues are of course dependent on $v$; their dependence will appear explicitly as needed.

\begin{lemma}\label{fixedpointestimates} There exists positive constants $C_\phi$, $C_{\lambda}$, and  $\varepsilon_0$  independent of
$\nvec$  such that for  two continuous potentials  $v_1$ and $v_2$ with small norms, 
$\|v_1\|_{\mathcal{C}}$ and $ \,\,\|v_2\|_{\mathcal{C}}\leq \varepsilon_o \sup_{n}|\widehat{\alpha}(n)|^2$, then: 
\begin{align}
  & (i) \,\,\,     \left\| P_{\nvec,\phi(x),r} (v_2) - P_{\nvec,\phi(x),r} (v_1)\right\|_{\mathcal{C}}\leq C_{\phi}\frac{ \|v_2-v_1\|_{\mathcal{C}} }{n^2+1} \\
  &(ii) \,\,\,  \left|\lambda_{\nvec}(v_2) -\lambda_{\nvec}(v_1)\right| \leq C_{\lambda}\frac{\|v_2-v_1\|_{\mathcal{C}}}{|n|+1}, \\
  &(iii)   \,\,\,\,   \left|\Re\left(\lambda_{\nvec}(v_2) -\lambda_{\nvec}(v_1)\right)\right|\leq C_{\lambda} \frac{\|v_2-v_1\|_{\mathcal{C}}}{n^2+1} 
\end{align}
\end{lemma}
\begin{proof}  (i) By writing the difference of projections for $\mA_{v_1}$ and $\mA_{v_2}$ as contour integrals using the second resolvent equation,
 the contour just about the eigenvalues $\lambda_{\nvec}(v_1)$ and $\lambda_{\nvec}(v_2)$ (the $v$'s are $\pi,\phi$ entries in $\mA_v$), we have
\begin{align}\label{PminusP}
P_{\nvec,\phi(x),r} (v_2) - P_{\nvec,\phi(x),r} &(v_1)= -\!\!\!\sum_{\mvec: \,\,\mvec\neq \nvec}\frac{P_{\nvec,\phi(x),\pi}(v_2) (v_2-v_1)P_{\mvec,\phi,r}(v_1)}{\lambda_{\nvec}(v_2)-\lambda_{\mvec}(v_1)} \nonumber\\
        & -\!\!\!\sum_{\mvec: \,\,\mvec\neq \nvec}\frac{
	P_{\mvec,\phi(x),\pi}(v_2)(v_2-v_1)P_{\nvec,\phi,r}(v_1)}{\lambda_{\nvec}(v_1)-\lambda_{\mvec}(v_2)}. 
\end{align}
 The first sum in Eq.(\ref{PminusP}) equals 
 \begin{align}
            \frac{e_{\nvec,\pi}(x,v_2)}{\lambda_{\nvec}\langle
	    f_{\nvec},e_{\nvec}\rangle_{\mathcal{H}}}\sum_{\mvec:
	    \,\,\mvec\neq \nvec}\frac{\langle f_{\nvec,\pi}(v_2),
	    (v_2-v_1)e_{\mvec,\pi}(v_1)\rangle\langle
	    f_{\mvec,\pi}(v_1),\alpha\rangle}{(\lambda_{\nvec}(v_2)-\lambda_{\mvec}(v_1))\lambda_{\mvec}(v_1)(\lambda_{\mvec}(v_1)+1)\langle
	    f_{\mvec},e_{\mvec}\rangle_{\mathcal{H}}}   
 \end{align}
which converges absolutely to a continuous function and is bounded by  ${\mathcal{O}}(n^{-2})\|v_2-v_1\|_{\mathcal{C}}$; in particular the nearly resonant
small denominator  $|\lambda_{\nvec}-\lambda_{\mvec}|\sim \|\widehat{\alpha}(n)\|^2/n$ is mollified by the other factors $\lambda_{\nvec}\lambda_{\mvec}(\lambda_{\mvec}+1)$
in the denominator.    The other sum in (\ref{PminusP}) is controlled similarly.  

 (ii,iii)  These estimates are  shown by first order perturbation theory, applied to the eigenvalue problem
                                   \begin{align}
                                        \left( \partial_x^2 -1 -\left(v_1 +\tau (v_2(x)-v_1(x))\right) -\frac{\lambda_{\nvec}(\tau)}{\lambda_{\nvec}(\tau)+1}|\alpha\rangle\cdot\langle \alpha|\right) e_{\nvec,\pi}(\tau,x)\nonumber\\
                                         = \lambda^2_{\nvec}(\tau)e_{\nvec,\pi}(\tau,x)
                                   \end{align}
The equation for $d \lambda_{\nvec}/d\tau$ is
\begin{align}
  &\left(2\lambda_{\nvec}(\tau) +  \frac{\langle f_{\nvec,\pi}(\tau),\alpha\rangle\cdot\langle\alpha,e_{\nvec,\pi}(\tau)\rangle}{\lambda^2_{\nvec}(\tau)
  \langle f_{\nvec,\pi}(\tau),e_{\nvec,\pi}(\tau)\rangle} \right)\frac{d\lambda_{\nvec}(\tau)}{d\tau}\nonumber\\
 & =  -\frac{\langle f_{\nvec,\pi}(\tau),(v_2-v_1)e_{\nvec,\pi}(\tau)\rangle}{\langle f_{\nvec,\pi}(\tau),e_{\nvec,\pi}(\tau)\rangle}\nonumber\\
 & = -\frac{\langle e_{\nvec,\pi}(\tau),
 (v_2-v_1)e_{\nvec,\pi}(\tau)\rangle}{\langle
 e_{\nvec,\pi}(\tau),e_{\nvec,\pi}(\tau)\rangle} +{\mathcal{O}}(n^{-1})\|v_2-v_1\|_{\mathcal{C}}.   
\end{align}
by Lemma(\ref{enalphadeltax}), see also the remark Ineq.(\ref{projectiondifference}) following it, with $\varepsilon= \lambda_{\nvec}/(\lambda_{\nvec}+1)$
having an imaginary part ${\mathcal{O}}(n^{-1})$. This shows that 
\begin{align}
\frac{d\lambda_{\nvec}(\tau)}{d\tau}=  -\frac{\langle
e_{\nvec,\pi}(\tau)(v_2-v_1)e_{\nvec,\pi}(\tau)\rangle}{2\lambda_{\nvec}(\tau)
\langle e_{\nvec,\pi}(\tau),e_{\nvec,\pi}(\tau)\rangle} 
+{\mathcal{O}}(n^{-2})\|v_2-v_1\|_{\mathcal{C}},
\end{align}
which is ${\mathcal{O}}(n^{-1})$ with real part, ${\mathcal{O}}(n^{-2})$.
Integrating this expression in $\tau$ from $0$ to $1$ gives (ii) and
(iii) of the lemma.
\end{proof}

\end{section}


\end{document}